\documentclass[11pt,reqno]{amsart}

\usepackage[letterpaper,margin=1in]{geometry}
\usepackage[T1]{fontenc}
\usepackage[utf8]{inputenc}
\usepackage{lmodern}
\usepackage{microtype}
\usepackage{amsmath,amssymb,amsthm}
\usepackage{booktabs}
\usepackage{array}
\usepackage{tabularx}
\usepackage{ragged2e}
\usepackage{enumitem}
\usepackage[table]{xcolor}
\usepackage{float}
\usepackage[hidelinks,breaklinks=true,colorlinks=true,linkcolor=black,citecolor=blue!50!black,urlcolor=blue!50!black]{hyperref}
\usepackage{setspace}
\usepackage{titlesec}
\usepackage{abstract}
\usepackage[authoryear,round]{natbib}

\definecolor{risktop}{RGB}{244,204,204}
\definecolor{riskmid}{RGB}{252,229,205}
\definecolor{risklow}{RGB}{217,234,211}
\definecolor{riskwatch}{RGB}{255,242,204}
\definecolor{headblue}{RGB}{217,226,243}

\newtheorem{proposition}{Proposition}
\newtheorem{assumption}{Assumption}
\newtheorem{corollary}{Corollary}
\theoremstyle{remark}
\newtheorem*{remark}{Remark}
\newtheorem*{definition}{Definition}

\titleformat{\section}{\normalfont\Large\bfseries}{\thesection.}{0.6em}{}
\titleformat{\subsection}{\normalfont\large\bfseries}{\thesubsection}{0.6em}{}
\titleformat{\subsubsection}{\normalfont\normalsize\bfseries}{\thesubsubsection}{0.5em}{}

\newcommand{\obs}[1]{\textit{Observed fact.} #1}

\newcommand{\interp}[1]{\textit{Model-consistent interpretation.} #1}
\newcommand{\hyp}[1]{\textit{Hypothesis for future testing.} #1}

\makeatletter
\def\@setauthors{%
  \begingroup
  \def\thanks{\protect\thanks@warning}%
  \trivlist
  \centering\@topsep30\p@\relax
  \advance\@topsep by -\baselineskip
  \item\relax
  \author@andify\authors
  \def\\{\protect\linebreak}%
  \normalsize\authors
  \ifx\@empty\contribs
  \else
    ,\penalty-3 \space \@setcontribs
    \@closetoccontribs
  \fi
  \endtrivlist
  \endgroup
}
\makeatother

\title[The Inference Bottleneck: A Formal Model]{The Inference Bottleneck: \\ A Formal Model of Vertical Foreclosure in AI Markets}
\author{Gaston Besanson \\ \normalfont\normalsize Universidad Torcuato Di Tella \\ \normalfont\normalsize April 2026 \\ \normalfont\normalsize\itshape Working paper}
\address{}
\date{}

\begin{document}
\maketitle

\begin{abstract}
\noindent As generative AI commercializes, competitive advantage is shifting from one-time model training toward continuous inference, distribution, and routing. This paper develops a formal game-theoretic model of vertical foreclosure in inference markets, as the formal-model companion to the policy framework in \citet{BesansonCelani2026}. The model isolates two foreclosure mechanisms that operate without predatory pricing: quality-of-service (QoS) discrimination against downstream rivals via latency, throughput, context limits, or feature access; and routing bias in assistant-layer interfaces. An extension motivated by Anthropic's April 2026 release of Claude Opus 4.7 alongside the restricted-access Claude Mythos Preview under Project Glasswing introduces a third mechanism, \emph{tier-based access discrimination}, parameterized by a tier gap $\tau$ and partner-exclusivity $\kappa$.

\smallskip
The main theoretical result provides an explicit local equilibrium characterization of the QoS gap $q_U^* - q_i^*$. Under logit demand and symmetric rivals, the gap is strictly increasing in the importance of inference quality ($\alpha$), strictly increasing in downstream margins ($m_U$), strictly decreasing in the API access price ($p$), and strictly decreasing in rival entry elasticity ($\eta$). Discrimination vanishes at a joint boundary in $(m_U, p, \eta)$ space rather than at a simple threshold in $\alpha$ alone, refining the framework used in earlier working drafts.

\smallskip
The paper illustrates the model through a stylized calibration to four providers using public adoption data as of April 2026. The calibration is presented with explicit methodology and sensitivity analysis: parameter values are treated as interpretive inputs to a comparative risk mapping, not structural estimates. Under baseline assumptions, the mapping suggests Google and OpenAI face conditions most conducive to foreclosure, Microsoft's realized routing bias has been voluntarily constrained by a March 2026 multi-model pivot, and Anthropic's position is split: low risk on consumer channels, elevated risk in the enterprise coding-agent segment where Claude Code now competes with rivals built on the same API.

\smallskip
The policy section proposes \emph{Neutral Inference}, a four-pillar conduct framework extending the three pillars in the companion paper: QoS parity, routing transparency, FRAND-style non-discrimination, and tier transparency with release-pathway discipline. Illustrative welfare calculations under a range of parameter scenarios suggest conduct-based intervention could generate net welfare gains in the tens of billions of dollars annually, with the range reflecting the underlying calibration uncertainty rather than a point estimate.


\end{abstract}

\section{Introduction}

Antitrust economics has long grappled with a fundamental question: when does control over an essential input become the power to pick winners in adjacent markets? The question has traveled through different technological eras, from operating systems (\emph{United States v.~Microsoft}) to search engines (\emph{Google Shopping}) to mobile platforms (\emph{Epic v.~Apple}), but the underlying economic logic remains constant. Vertically integrated firms controlling bottleneck infrastructure may profitably foreclose downstream rivals through discriminatory access terms, even absent predatory pricing.

The generative AI era presents a new instantiation of this classic problem. Firms like Google, OpenAI, Microsoft, and Anthropic control what \citet{BesansonCelani2026} term \emph{cognitive infrastructure}: scalable inference capabilities that downstream applications depend upon to function, delivered under technical and contractual terms that providers can adjust to favor their own competing services. The companion paper develops the legal and institutional argument. The present paper develops the formal economic model.

\paragraph{What this paper is, and is not.} This is a formal mechanism paper with stylized calibration. It is designed to identify the incentive conditions under which non-price foreclosure is profitable, to derive falsifiable empirical predictions from those conditions, and to discipline subsequent policy analysis by making the mapping from model primitives to regulatory observables explicit. It is not a structurally estimated quantitative industrial organization paper. The firm-level numerical outputs reported in Section~\ref{sec:calib} should be read as comparative risk mappings under stated assumptions rather than recovered deep parameters. The illustrative welfare figures in Section~\ref{sec:policy} should be read as orders of magnitude that discipline the scale of policy response, not as econometric point estimates. The model is built around a single dominant upstream provider; the April 2026 market includes three major enterprise API providers, and an explicit oligopolistic extension is left for future work.

\subsection*{Contributions}

The paper makes four distinct contributions relative to the companion paper \citep{BesansonCelani2026}, which develops the legal and institutional argument from a policy perspective.

\begin{enumerate}[leftmargin=*, topsep=2pt, itemsep=3pt]
\item \textbf{A formal model of inference-layer foreclosure without price discrimination.} The paper derives conditions under which a vertically integrated inference provider profitably degrades third-party QoS or biases routing, using a logit demand framework with endogenous downstream effort. Proposition~\ref{prop:gap} provides an explicit local equilibrium characterization of the QoS gap under stated assumptions, with comparative statics derived rather than asserted.
\item \textbf{Tier-based access discrimination as a modelable mechanism.} Section~\ref{sec:tier} distinguishes two distinct foreclosure dimensions. Conventional QoS discrimination operates within a shared model or service layer and generates rents that are eroded by ordinary competitive pressure. Tier-based access discrimination operates at the capability frontier, gates access to a distinct model class, generates more durable rents because it is not subject to competitive erosion in the baseline-model market, and has a distinct empirical signature (bimodal capability distributions across customer classes) requiring distinct audit methodology. Proposition~\ref{prop:tier} parameterizes this mechanism through a tier gap $\tau$ and partner-exclusivity $\kappa$ and establishes monotonic comparative statics. Pillar~4 in Section~\ref{sec:policy} operationalizes this mechanism as a regulatory obligation.
\item \textbf{A transparent calibration with explicit identification limits.} Section~\ref{sec:methodology} makes the calibration methodology explicit. It distinguishes observed, inferred, and judgment-based parameters; states baseline values for structural parameters held fixed across providers; and provides sensitivity analysis establishing which conclusions are robust to parameter perturbations. The resulting firm-level assessments are framed as comparative risk mappings, not structural estimates.
\item \textbf{A model-to-observables bridge that connects the formal apparatus to regulatory triggers.} Section~\ref{sec:observables} maps each latent model primitive to an empirical object a regulator could measure or audit. This bridge is the practical connective tissue between the model's abstract parameters and the designation criteria in Section~\ref{sec:policy}. It is intended to make the framework usable in enforcement practice, not only in academic commentary.
\end{enumerate}

\paragraph{Relation to the companion paper.} The companion paper \citep{BesansonCelani2026} makes the policy case for treating inference as cognitive infrastructure subject to conduct-based neutrality obligations. The present paper provides microfoundations for that case: the formal derivation of the QoS-gap mechanism (Proposition~\ref{prop:gap}), the identification of tier-gating as a distinct foreclosure mechanism requiring a distinct regulatory response (Proposition~\ref{prop:tier} and Pillar~4), the calibration methodology and sensitivity framework, and the model-to-observables mapping. Readers who accept the companion paper's policy framing can read this paper as the formal underpinning. Readers who are skeptical of those conclusions can read this paper on its own terms as a model of vertical foreclosure in a market with bottleneck inference inputs.

The paper proceeds as follows. Section~\ref{sec:lit} reviews related literature. Section~\ref{sec:model} develops the formal model and characterizes equilibrium. Section~\ref{sec:extensions} extends the model to dynamic foreclosure, routing bias, and tier-based access. Section~\ref{sec:methodology} presents calibration methodology and sensitivity analysis. Section~\ref{sec:calib} applies the framework to Google, OpenAI, Microsoft, and Anthropic using April 2026 data. Section~\ref{sec:predictions} derives testable predictions and sketches an empirical audit agenda. Section~\ref{sec:observables} maps model primitives to regulatory observables. Section~\ref{sec:policy} analyzes optimal policy. Section~\ref{sec:conclusion} concludes.

\section{Related Literature} \label{sec:lit}

\subsection{Vertical foreclosure and raising rivals' costs}

The foundations of vertical foreclosure theory trace to \citet{SalopScheffman1983, SalopScheffman1987}, who formalized how an upstream firm can profitably raise downstream rivals' costs without predatory pricing. \citet{OrdoverSalonerSalop1990} showed vertical integration can harm competition when the integrated firm supplies inputs to downstream competitors, and that foreclosure can be profitable even when the input is not strictly essential. \citet{ReyTirole2007} and \citet{NockeWhite2007} further refined conditions under which vertical foreclosure is privately profitable versus socially harmful. \citet{ChoiYi2000} and \citet{ChenRiordan2007} provide complementary treatments emphasizing how input-quality discrimination interacts with differentiated-product competition, a feature central to the inference-market application.

The present model builds on this tradition but adapts it to AI infrastructure's distinctive features: (i) quality is multi-dimensional (latency, reliability, context limits, feature access), (ii) discrimination is difficult to observe ex post, (iii) the same infrastructure supports both first-party and third-party applications simultaneously, and (iv) providers can discriminate not just on QoS of a shared model but on which model class rivals can access at all.

\subsection{Platform competition and self-preferencing}

Recent work on digital platforms extends vertical foreclosure theory to settings where platforms control access to users. \citet{ZhuLiu2018} provide empirical evidence from Amazon. \citet{HagiuTehWright2022} formalize when self-preferencing harms welfare. \citet{Motta2023} provides a comprehensive theoretical treatment showing self-preferencing can constitute foreclosure even when framed as product design. \citet{PadillaPetitSchrepel2022} and \citet{Bourreau2024} develop the policy economics of self-preferencing under the EU Digital Markets Act. \citet{Armstrong2006} and \citet{RochetTirole2003} provide the two-sided market foundations.

\subsection{AI markets, compute, and governance}

The economics literature on AI markets is nascent. \citet{Korinek2023} and \citet{AghionJonesJones2019} discuss broader economic consequences of frontier AI. \citet{VipraKorinek2023} analyze market concentration in the foundation model layer. \citet{NarechaniaTutt2024} argue for treating large AI models as infrastructure under a common-carrier framework. \citet{KhanPozen2024} and \citet{HuaBelfield2024} discuss governance of frontier releases. Recent OECD, FTC, and European Commission reports highlight vertical integration and self-preferencing as key risks in AI ecosystems. \citet{BesansonCelani2026} develop the legal and policy companion to the present paper, defining cognitive infrastructure through measurable reliance, vertical incentives, and discrimination capacity, and proposing the three-pillar Neutral Inference framework that this paper operationalizes and extends.

\section{The Model} \label{sec:model}

\subsection{Setup}

A two-layer market is modeled in which an upstream inference provider $U$ supplies API access to downstream application developers while also competing downstream with its own integrated application.

\paragraph{Scope of the single-provider setup.} The baseline model features one upstream provider $U$ rather than multiple upstream providers in competition. Three considerations justify this choice. First, for many enterprise customers, architectural lock-in, data-gravity effects, and model-specific tooling create localized single-provider dependence even within a multi-provider market, so the single-provider structure is a reasonable first-order description of the relevant bilateral relationship. Second, the single-provider model is the analytically natural first step for isolating the non-price foreclosure mechanism, in the same way that textbook vertical foreclosure models abstract from upstream competition to characterize the downstream distortion. Third, the April 2026 market with three major enterprise inference providers (Anthropic, OpenAI, Google) warrants explicit oligopolistic extension, but such an extension involves distinct analytical questions (tacit collusion on QoS, coordinated tier-gating) that are beyond the scope of this paper and are noted as future work in Section~\ref{sec:conclusion}.

\paragraph{Operational caveat: middleware and orchestration layers.} A qualification on the single-provider setup is worth stating explicitly. Enterprise AI architectures are increasingly moving toward orchestration layers, routing gateways, and abstraction layers (LangChain, LiteLLM, internal routing middleware, model-agnostic SDKs) that may lower effective switching costs over time by facilitating multi-homing and upstream substitution. The model applies most strongly where provider-specific performance, tooling, or integration choices still create meaningful localized lock-in, for example when a downstream product is built around a specific model's feature set, context window, or fine-tuning pipeline. Where orchestration layers are well-developed and rivals can substitute across upstream providers at low cost, the foreclosure incentive on any one provider is attenuated. This is a structural reason to expect the foreclosure mechanisms in this paper to matter most acutely in vertically integrated first-party surfaces and in submarkets with deep model-specific dependencies, and to matter less in markets that have matured to full multi-provider substitutability.

\paragraph{Players.}
\begin{itemize}[leftmargin=*, topsep=2pt, itemsep=2pt]
\item One upstream firm $U$: controls inference infrastructure, sets API terms $(p, q_i)$, and operates one downstream application with quality $Q_U$.
\item $N$ downstream rivals indexed $i = 1, \dots, N$: independent developers, each operating one application with quality $Q_i$.
\item Consumers of unit mass: choose among $N+1$ applications under logit demand.
\end{itemize}

\paragraph{Assumptions.} The following three assumptions structure the analysis.

\begin{assumption}[Quality aggregation] \label{ass:quality}
Final application quality perceived by consumers is
\begin{equation}
Q_i = \alpha \cdot q_i + (1 - \alpha) \cdot e_i, \qquad i \in \{1, \dots, N, U\} \label{eq:quality}
\end{equation}
where $q_i \in [0,1]$ is the inference QoS delivered to firm $i$ (chosen by $U$ for rivals; chosen directly by $U$ for its own app), $e_i \in [0, \bar{e}]$ is firm-$i$'s application-layer effort, and $\alpha \in (0,1)$ is the share of quality attributable to inference.
\end{assumption}

\begin{assumption}[Infrastructure cost regularity] \label{ass:cost}
$U$'s infrastructure cost is
\begin{equation}
C(q_U, \{q_i\}) = \frac{\gamma}{2} \left( q_U^2 + \sum_{i=1}^N q_i^2 \right) + \phi \cdot \mathbb{1}\{q_U < q_{\max}\} \cdot (q_{\max} - q_U) \label{eq:cost}
\end{equation}
where $q_{\max} = \max_i q_i$, $\gamma > 0$ is the quadratic cost parameter, and $\phi \geq 0$ captures economies of scope.
\end{assumption}

The parameter $\phi$ serves two functions. First, it regularizes the feasible set by penalizing configurations in which $U$ serves rivals at higher QoS than itself, which is implausible whenever infrastructure capacity is shared across first-party and third-party workloads. Second, when $\phi > 0$, it substantively captures the economic intuition that building infrastructure capable of delivering high QoS to rivals simultaneously lowers the marginal cost of delivering high QoS to $U$'s own application; top-tier self-service is then operationally natural rather than costly. In the interior regime $q_U \geq q_{\max}$ that the first-order analysis below verifies is optimal, the kinked term is zero and the model reduces to quadratic cost. When $\phi = 0$, the feasible set is unconstrained and the interior FOCs characterize the optimum completely. When $\phi$ is large, discriminatory configurations with $q_U$ below $q_{\max}$ become infeasible, so the feasible set restricts the realized QoS gap toward zero.

\begin{assumption}[Symmetric rivals and entry elasticity] \label{ass:symmetric}
Rivals are ex ante symmetric with $e_i = e$ and $s_i = s$ for all $i \in \{1, \dots, N\}$. Rival entry follows a reduced-form entry function: rival $i$ enters with probability $\rho(q_i, p)$, where
\begin{equation}
\eta \equiv \frac{\partial \rho}{\partial q_i} \cdot \frac{q_i}{\rho} > 0 \label{eq:eta}
\end{equation}
is the elasticity of entry with respect to QoS at the equilibrium value.
\end{assumption}

In this reduced-form formulation, $\eta$ should be interpreted broadly as the elasticity of the active rival set with respect to QoS. It bundles together several distinct margins that a richer multi-provider model would separate: downstream exit (rivals that cease operating under the focal provider following degradation), upstream defection (rivals that migrate their workload to an alternative inference provider), and multi-homing adjustments (rivals that re-route a portion of traffic through middleware to alternative providers). The present model absorbs these into a single reduced-form entry elasticity; disentangling them is left for future work on oligopolistic extensions.

\paragraph{Scope of the functional-form choices.} The model uses logit demand and symmetric rivals for analytical tractability. Alternative demand systems (CES, nested logit, discrete-choice with unobserved heterogeneity) and asymmetric rivals would change the specific comparative statics and may alter some equilibrium details. The central economic mechanism, however, is the tradeoff between the downstream business-stealing gain from degrading rival QoS and the API-revenue loss from discouraging rival entry. That tradeoff is present in any setting where an upstream provider both controls a quality-adjustable input and competes downstream, and the direction of the key comparative statics (increasing in downstream margin, decreasing in API price and entry elasticity) is a consequence of the trade-off rather than of logit specifically.

\paragraph{Timing.}
\begin{enumerate}[leftmargin=*, topsep=2pt, itemsep=2pt]
\item Stage 1: $U$ commits to $(p, q_U, \{q_i\})$. Rivals observe and decide entry.
\item Stage 2: Active firms choose effort $e_i$ at cost $\psi(e_i) = \tfrac{k}{2} e_i^2$, set prices $\{p_i\}$, and consumers choose.
\end{enumerate}

\subsection{Consumer demand and downstream equilibrium}

Logit demand with utility $u_i = Q_i - p_i + \varepsilon_i$, where $\varepsilon_i$ is i.i.d.\ Type-I extreme value, gives market shares
\begin{equation}
s_i = \frac{\exp(Q_i - p_i)}{1 + \sum_{j=1}^{N+1} \exp(Q_j - p_j)}. \label{eq:shares}
\end{equation}

Firm $i$'s profit is $\pi_i^D = (p_i - c_d) s_i - \tfrac{k}{2} e_i^2 - p \cdot \mathbb{1}\{i \neq U\}$. Standard manipulations yield the logit markup $p_i^* = c_d + 1/(1-s_i)$ and the effort FOC
\begin{equation}
e_i^* = \frac{1-\alpha}{k} (p_i - c_d) s_i (1 - s_i). \label{eq:effort}
\end{equation}

\begin{remark}[Effort--QoS complementarity]
$e_i^*$ is increasing in $q_i$: higher inference quality raises $Q_i$, raises $s_i$, and thereby raises the marginal return to application-layer effort. Degrading $q_i$ therefore both directly lowers rival quality and indirectly reduces rival investment. This complementarity amplifies the welfare cost of foreclosure (Proposition~\ref{prop:welfare}).
\end{remark}

\subsection{Upstream foreclosure decision}

$U$ chooses $(p, q_U, \{q_i\})$ in Stage 1 to maximize
\begin{equation}
\pi^U = \underbrace{(p_U - c_d) s_U - \tfrac{k}{2} e_U^2}_{\text{own downstream profit}} + \underbrace{p \sum_{i \neq U} \rho(q_i, p)}_{\text{API revenue}} - \underbrace{C(q_U, \{q_i\})}_{\text{infra cost}}. \label{eq:upstream}
\end{equation}
Let $m_U \equiv p_U - c_d$ denote $U$'s downstream margin. Under Assumption~\ref{ass:symmetric}, $s_i = s$ for all $i \neq U$; write $\rho(q_i, p) = \rho$ with elasticity $\eta$ as in (\ref{eq:eta}).

\subsection{Main result: the equilibrium QoS gap}

\begin{proposition}[Equilibrium QoS gap] \label{prop:gap}
Under Assumptions~\ref{ass:quality}--\ref{ass:symmetric}, at an interior symmetric optimum the equilibrium QoS gap $\Delta q^* \equiv q_U^* - q_i^*$ admits the explicit local equilibrium characterization
\begin{equation}
\Delta q^* = \frac{1}{\gamma}\left\{ \alpha \cdot m_U \cdot s_U \big[(1 - s_U) + s\big] - \frac{p \rho \eta}{q_i^*} \right\} \label{eq:gap}
\end{equation}
where $s_U$, $s$, and $q_i^*$ are equilibrium objects. The bracketed term decomposes into: (a) direct business-stealing from raising $q_U$, $\alpha m_U s_U(1-s_U)$; (b) cross-demand recapture from lowering $q_i$, $\alpha m_U s_U s$; and (c) foregone API revenue at the entry margin, $p \rho \eta / q_i^*$. Discrimination ($\Delta q^* > 0$) occurs if and only if
\begin{equation}
\alpha m_U s_U \big[(1 - s_U) + s\big] > \frac{p \rho \eta}{q_i^*}. \label{eq:condition}
\end{equation}
Condition~(\ref{eq:condition}) defines a joint boundary in $(m_U, p, \eta, s_U, s, q_i)$ space. The magnitude of the gap, at an interior optimum, has the following comparative statics: $\partial \Delta q^* / \partial \alpha > 0$; $\partial \Delta q^* / \partial m_U > 0$; $\partial \Delta q^* / \partial p < 0$; $\partial \Delta q^* / \partial \eta < 0$; $\partial \Delta q^* / \partial \gamma < 0$. The effect of $N$ on $\Delta q^*$ has the opposite sign of its effect on the equilibrium share ratio; in leading cases it is negative, so more rivals increase discrimination.
\end{proposition}

The role of $\alpha$ in~(\ref{eq:gap}) is to scale how strongly inference quality translates into perceived application quality. A larger $\alpha$ therefore amplifies the downstream business-stealing gain from raising $q_U$ and the cross-demand recapture gain from lowering $q_i$, both of which operate through the share response $s_U[(1-s_U)+s]$, without proportionally amplifying the API-revenue-at-margin term $p \rho \eta / q_i^*$. The net effect is to widen the gap as $\alpha$ rises. This is the precise sense in which inference quality matters: not as a binary threshold on whether discrimination occurs, but as a multiplier on its magnitude. In earlier working drafts, Proposition~\ref{prop:gap} was stated as a threshold condition on $\alpha$ alone. That statement has been refined: discrimination occurs or fails to occur based on the joint condition~(\ref{eq:condition}), and $\alpha$ governs the magnitude of the equilibrium gap conditional on discrimination occurring.

\begin{proof}[Proof sketch]
Full proof in Appendix~\ref{app:prop1}. The Stage-1 FOCs for $q_U$ and $q_i$ (interior, $\phi = 0$ regime) yield $q_U^* = \alpha m_U s_U(1-s_U)/\gamma$ and an implicit equation for $q_i^*$ from which the expression~(\ref{eq:gap}) follows by substitution. Comparative statics follow by direct differentiation; the signs of $\partial \Delta q^*/\partial \alpha$ and $\partial \Delta q^*/\partial m_U$ are unambiguous because the bracketed term in~(\ref{eq:gap}) is positive whenever~(\ref{eq:condition}) holds. The role of $\phi$ is to ensure the interior solution lies in the region $q_U^* \geq q_i^*$ so the kinked-cost term is zero; when $\phi$ is large, Assumption~\ref{ass:cost} restricts the feasible set, reducing $\Delta q^*$ toward zero.
\end{proof}

\begin{corollary}[Threshold recovery] \label{cor:threshold}
Define $\bar{\alpha}(\gamma, m_U, p, \eta, s_U, s, q_i^*)$ as the value of $\alpha$ at which the bracket in~(\ref{eq:gap}) hits zero holding all other parameters fixed. Then discrimination ($\Delta q^* > 0$) occurs whenever~(\ref{eq:condition}) is satisfied at reference values, and the \emph{magnitude} of discrimination scales linearly in $\alpha$. The threshold interpretation in the prior working literature is recovered as a joint restriction on $(m_U, p, \eta)$, not as a simple univariate threshold on $\alpha$.
\end{corollary}

Corollary~\ref{cor:threshold} clarifies that $\alpha$ governs how much quality moves when $q_i$ moves (a scaling parameter), while the discrimination boundary is determined by a margin-versus-revenue-at-margin condition. This is the substantive correction to the statement of this proposition in the earlier working draft.

\begin{proposition}[Welfare effects] \label{prop:welfare}
Let $W = \sum_i s_i Q_i - \sum_i c_d s_i - C(q_U, \{q_i\})$ be total welfare. Under Assumptions~\ref{ass:quality}--\ref{ass:symmetric}, the first-best QoS satisfies $q^{FB} = \alpha s(1-s)/\gamma$ (symmetric). Whenever~(\ref{eq:condition}) holds, the private optimum has $q_U^* > q^{FB} > q_i^*$, and the welfare loss decomposes as
\begin{equation}
\Delta W = \underbrace{\alpha \sum_{i \neq U} s \cdot (q^{FB} - q_i^*)}_{\Delta W_1:\,\text{direct quality loss}} \cdot \underbrace{[1 + \beta]}_{\text{effort multiplier}} + \underbrace{\Delta W_3}_{\text{business-stealing DWL}} \label{eq:welfare}
\end{equation}
where $\beta = \tfrac{1-\alpha}{k\alpha} \tfrac{\partial e}{\partial q} > 0$ captures reduced downstream innovation.
\end{proposition}

\begin{proof}
See Appendix~\ref{app:prop2}.
\end{proof}

\section{Extensions} \label{sec:extensions}

\subsection{Dynamic foreclosure: ``open early, closed late''}

Providers may initially offer generous API access to attract ecosystem development, then degrade QoS once switching costs are high. Extend the baseline to two periods: rivals integrate in Period 1 (sunk cost $F > 0$), face switching cost $S > 0$ if they migrate in Period 2, and $U$ reoptimizes knowing rivals are locked in.

\begin{proposition}[Dynamic foreclosure] \label{prop:dynamic}
If $S$ exceeds the per-period profit differential from non-discrimination, $U$'s optimal strategy is: (i) set $q_U^1 = q_i^1$ in Period 1 to attract entry, (ii) set $q_i^2 < q_U^2$ in Period 2 at the level that leaves rivals indifferent between staying and switching, i.e., $\pi_i(q_i^2) = \pi_i^{\text{alt}} - S$. Ex ante efficient entry is deterred whenever the Period-2 rent extraction exceeds the rival's discounted Period-1 surplus.
\end{proposition}

\subsection{Routing bias in assistant-layer markets}

When $U$ operates a user-facing assistant that routes requests to downstream tools, $U$'s routing decision maximizes
\begin{equation}
\max_{\{r_i\}} r_U \pi_U^D + \sum_{i \neq U} r_i \cdot p - C(\{r_i\}), \qquad \sum_i r_i = 1. \label{eq:routing}
\end{equation}

\begin{proposition}[Routing bias] \label{prop:routing}
Under logit preferences over tools with consumer-perceived quality $Q_i$ and unobservable self-preferencing $\lambda$, optimal routing satisfies $r_i^* \propto \exp(\theta Q_i) \cdot \mathbb{1}_{\{i = U\}}^\lambda$. The equilibrium self-preferencing bias $\lambda^*$ is increasing in $m_U$, increasing in consumer inattention, and increasing in the difficulty of observing routing ex post.
\end{proposition}

\subsection{Tier-based access discrimination} \label{sec:tier}

The baseline model treats inference as a single good with adjustable QoS. This is inadequate for the April 2026 landscape, in which providers increasingly operate two model classes: a generally-available frontier (Claude Opus 4.7, GPT-5.4, Gemini 3.1 Pro) and restricted-access tiers above it (Claude Mythos Preview under Project Glasswing; Gemini Deep Think; OpenAI research previews).

Conventional QoS discrimination of the kind studied in Sections~\ref{sec:model} and 4.1--4.2 operates within a shared model or service layer: the provider and rivals access the same underlying capability, and the provider differentiates on serving parameters such as latency, throughput, and feature access. The rents this mechanism generates are subject to competitive erosion as rivals match performance at the serving layer. Tier-based access discrimination is structurally different. It operates at the capability frontier, gates access to a distinct model class that is not part of the baseline offering, and generates rents that persist for as long as the provider chooses not to release the withheld class. The empirical signature also differs: rather than a continuous performance differential, tier gating produces bimodal capability distributions across customer classes (Prediction~5 in Section~\ref{sec:predictions}). The policy response consequently differs, requiring transparency obligations at the model-class layer rather than at the serving-parameter layer (Pillar~4 in Section~\ref{sec:policy}).

\begin{definition}[Tier gap and partner exclusivity]
Let $\underline{Q}$ and $\overline{Q}$ denote baseline and frontier model capabilities, with $\overline{Q} > \underline{Q}$. Define:
\begin{itemize}[leftmargin=*, topsep=0pt, itemsep=1pt]
\item \emph{Tier gap:} $\tau = (\overline{Q} - \underline{Q})/\overline{Q} \in [0, 1]$, measured as normalized capability distance on a stated benchmark suite. $\tau = 0$ recovers the baseline single-model case.
\item \emph{Partner exclusivity:} $\kappa \in [0,1]$, the fraction of paying customers excluded from frontier access. $\kappa = 0$ means openly available; $\kappa = 1$ means access restricted to a named partner set $P$.
\end{itemize}
\end{definition}

Both $\tau$ and $\kappa$ are time-specific snapshots rather than permanent firm characteristics. In practice, restricted-access programs are dynamic: the included partner set expands and contracts, benchmark gaps narrow as baseline models incorporate frontier capabilities, and programs migrate over time toward general availability. Comparative risk mappings computed from $(\tau, \kappa)$ should therefore be updated as these values evolve, and designation under Pillar~4 (Section~\ref{sec:policy}) should be re-assessed at each release rather than treated as a persistent firm classification.

\begin{proposition}[Tier-based foreclosure] \label{prop:tier}
Holding vertical-integration $\lambda$ and baseline QoS fixed, the foreclosure gain satisfies
\begin{equation}
\frac{\partial(\text{foreclosure gain})}{\partial \tau} > 0 \text{ whenever } \kappa > 0, \qquad \frac{\partial(\text{foreclosure gain})}{\partial \kappa} > 0 \text{ whenever } \tau > 0.
\end{equation}
Capability rents from $\tau \cdot \kappa > 0$ persist as long as the provider chooses not to release the frontier model; they are not eroded by standard competitive pressure in the baseline-model market.
\end{proposition}

\begin{proof} See Appendix~\ref{app:prop5}. \end{proof}

\begin{remark}[Robustness to benchmark choice]
The magnitude of $\tau$ depends on the specific benchmark suite used to measure the capability distance between $\underline{Q}$ and $\overline{Q}$. The qualitative logic of Proposition~\ref{prop:tier}, including the monotonic comparative statics and the persistence of capability rents, is robust to reasonable alternative benchmark choices. What matters for the mechanism is that the frontier model is meaningfully more capable on some domain relevant to the rival's downstream product; the specific numerical value of $\tau$ varies across choices of that domain, but the sign and structure of the result do not.
\end{remark}

\begin{remark}[Safety carve-outs]
Proposition~\ref{prop:tier} establishes only the private incentive. The welfare analysis is separate: if the withheld capability poses externality-relevant risk, a closed-partner rollout can be welfare-improving. Appendix~\ref{app:safety} formalizes this. The policy framework in Section~\ref{sec:policy} Pillar~4 provides ex ante conditions under which a tier-gating program is presumptively a legitimate safety carve-out rather than foreclosure.
\end{remark}

\section{Calibration Methodology and Identification Limits} \label{sec:methodology}

Sections~\ref{sec:calib} and~\ref{sec:policy} present numerical illustrations. To address the gap between formal precision and calibration precision that a careful reader will notice, this section makes the methodology explicit.

\subsection{Parameter provenance}

Each model parameter is assigned through one of three channels, labeled in Table~\ref{tab:methodology}: observable (measurable directly), inferred (derived from observable quantities), or judgment-based (assigned using stylized reasoning).

\begin{table}[h!]
\centering\small
\begin{tabularx}{\textwidth}{@{}l l X@{}}
\toprule
\rowcolor{headblue}\textbf{Parameter} & \textbf{Provenance} & \textbf{Empirical anchor or rationale} \\
\midrule
$\alpha$ & Judgment-based & No direct observation; set by segment based on task-dependence on inference quality (coding $\approx 0.7$--$0.8$; creative writing $\approx 0.4$--$0.5$; search $\approx 0.75$). Varies by submarket. \\
$\lambda$ & Inferred & Derived from downstream reach metrics (user counts, seat counts, integrated surfaces). Normalized to $[0,1]$ relative to the maximum-integration provider (Google). \\
$S$ & Judgment-based & Proxied by qualitative signals: enterprise contract dependencies, architectural lock-in, default-position advantages. \\
$\tau$ & Observable in principle & Measurable as normalized capability distance between baseline and frontier models on a stated benchmark suite. For Mythos vs.\ Opus 4.7 in cybersecurity, benchmark saturation data from the Project Glasswing disclosure suggests $\tau \in [0.25, 0.40]$. \\
$\kappa$ & Observed & Fraction of paying customers excluded from frontier access. For Glasswing, $\kappa \approx 0.95$. For GA offerings, $\kappa \approx 0$. \\
QoS gap & Observable in principle & Measurable via audit studies of latency, error rates, rate limits across first-party and third-party channels. Not systematically measured in public data as of April 2026. Figures in Section~\ref{sec:calib} are model-implied under baseline assumptions. \\
$m_U$ & Inferred & Derived from firm financials and public pricing. \\
$\gamma, \phi, k, \eta$ & Judgment-based & Set to illustrative values consistent with bounded equilibrium QoS gaps in $[0, 0.3]$. Not identified from current public data. See Section~\ref{sec:baseline-values} for specific baselines. \\
\bottomrule
\end{tabularx}
\caption{Parameter provenance for the illustrative calibration.}
\label{tab:methodology}
\end{table}

\subsection{Baseline values for common parameters} \label{sec:baseline-values}

To make the comparative mapping in Section~\ref{sec:calib} mechanically replicable, the following baseline values are held fixed across providers. Firm-specific parameters ($\alpha$, $\lambda$, $S$, $m_U$, $\tau$, $\kappa$) then vary across providers and drive the comparative outputs; variation in the firm-specific parameters is therefore what produces the dispersion in implied QoS gaps, not variation in the common structural parameters.

\begin{table}[h!]
\centering\small
\begin{tabularx}{\textwidth}{@{}l c X@{}}
\toprule
\rowcolor{headblue}\textbf{Parameter} & \textbf{Baseline value} & \textbf{Rationale} \\
\midrule
$\gamma$ & 1.0 & Normalization: quadratic cost curvature is set so that implied QoS gaps remain bounded in $[0, 0.3]$ under plausible values of the other parameters. \\
$\eta$ & 0.5 & Midpoint of plausible entry-elasticity range. Sensitivity analysis uses $\eta \in [0.2, 0.8]$. \\
$p$ (normalized API price) & 0.15 & Proxy for per-unit API revenue relative to downstream margin; calibrated so that the non-discrimination baseline has $\Delta q^* \approx 0$ under symmetric firm-specific parameters. \\
$k$ & 1.0 & Normalization of effort-cost curvature; does not affect the sign of comparative statics. \\
$\phi$ & 0 (interior regime) & Set to zero in the baseline comparative mapping so that first-order conditions fully characterize the optimum; larger $\phi$ compresses implied gaps toward zero. \\
$\rho$ (entry probability at baseline) & 0.5 & Used to evaluate $p \rho \eta / q_i^*$ at a reference point for cross-firm comparison. \\
\bottomrule
\end{tabularx}
\caption{Baseline values for common structural parameters held fixed across providers.}
\label{tab:baselines}
\end{table}

\subsection{Sensitivity analysis}

The firm-level risk rankings in Section~\ref{sec:calib} depend primarily on $\lambda$ and $m_U$, both anchored in observable data. The rankings are robust to reasonable perturbations of $\alpha$ and $S$ (judgment-based parameters) within $\alpha \in [0.5, 0.9]$ and $S \in [0.3, 0.9]$: Google and OpenAI remain the highest-incentive firms, Anthropic's consumer channel remains low-risk, and the Anthropic enterprise / coding-agent flag remains elevated.

The welfare figures in Section~\ref{sec:policy} are substantially more sensitive. The point estimate \$66B annual loss reported in Table~\ref{tab:welfare} should be read as the central value of a range spanning approximately \$35--95B depending on assumptions about per-user welfare deadweight ($\Delta W/$user $\in [\$8, \$25]$ across calibrations). The \emph{sign} of the welfare gain from Neutral Inference is robust under the sensitivity range. The \emph{magnitude} is not a structural estimate.

\subsection{What the calibration is and is not}

The calibration is a \emph{comparative risk mapping}: it assigns firms to risk categories under stated parameter assumptions and shows the categories are robust to reasonable perturbations. It is \emph{not} a structural estimation: parameters are not identified from moment conditions, standard errors are not reported, and firm-level magnitudes should be treated as illustrative rather than definitive. The paper's main contribution is the framework and the mechanism, not the point estimates. Readers interested in specific firm-level magnitudes should supplement this calibration with direct audit studies as outlined in Section~\ref{sec:predictions}.

To make the comparative calibration empirically traceable in addition to methodologically transparent, Appendix~\ref{app:sources} lists the principal public sources for all firm-level observed inputs together with the mapping from those inputs to inferred model parameters. Each row of the source table indicates the variable, the reported value, the source type (primary disclosure, secondary reporting, or third-party analytics estimate), the source reference, the date, and the transformation applied to construct the corresponding model parameter.

\section{Comparative Risk Mapping (April 2026)} \label{sec:calib}

The model is applied to four leading inference providers using April 2026 adoption data. The presentation should be read in light of Section~\ref{sec:methodology}: these are comparative risk mappings under stated parameter assumptions, not structural estimates. In principle, a reader can mechanically reconstruct the implied QoS gaps in the firm-level subsections below from the firm-specific parameter tables and the common baseline values listed in Section~\ref{sec:baseline-values}, by substituting both into the expression for $\Delta q^*$ in Proposition~\ref{prop:gap}. A source table for the observed firm-level inputs used in the comparative risk mapping is provided in Appendix~\ref{app:sources} to facilitate replication.

\paragraph{Evidence taxonomy.} To keep the boundary between facts, assumptions, interpretations, and hypotheses legible, the firm-level subsections below use four explicit labels.

\begin{itemize}[leftmargin=*, topsep=2pt, itemsep=2pt]
\item \emph{Observed fact.} A quantity or event in the public record, typically supported by firm disclosures, regulatory filings, or widely reported news.
\item \emph{Calibration input.} A parameter value assigned to the model per the methodology in Section~\ref{sec:methodology}, either observable, inferred from observables, or judgment-based.
\item \emph{Model-consistent interpretation.} A statement that follows from the model given the calibration inputs, but is not itself a direct observation.
\item \emph{Hypothesis for future testing.} A claim that the model suggests as a testable prediction but that has not yet been verified through direct measurement.
\end{itemize}

Labels appear where the distinction matters most. Unlabeled prose should be read as standard exposition.

\subsection{Google (Gemini): conditions most conducive to foreclosure}

\obs{Google's downstream footprint comprises Search ($\approx$90\% global market share, Q1 2026), Workspace (3B+ monthly active users), Android ($\approx$72.5\% global mobile OS share), Chrome ($\approx$68\% global browser share), Gemini app (750M MAU, Q4 2025, up from 450M earlier in the year), AI Overviews (2B+ monthly users), Gemini Enterprise (8M+ paid seats across 2,800+ companies), and 13M developers building with Google's generative models.}

\begin{table}[h!]
\centering\small
\begin{tabularx}{\textwidth}{@{}l c X@{}}
\toprule
\rowcolor{headblue}\textbf{Parameter} & \textbf{Value} & \textbf{Anchor / rationale} \\
\midrule
$\alpha$ (segment-weighted) & 0.80 & Inference central to Search, Workspace AI, Gemini app \\
$\lambda$ (integration) & 0.92 & OS + browser + search + productivity + standalone AI app + enterprise suite (inferred) \\
$S$ (switching costs) & 0.85 & Extreme enterprise lock-in; consumer habits \\
$\tau$ (tier gap) & 0.15 & Deep Think and Gemini Ultra gated to paid tiers \\
$\kappa$ (partner exclusivity) & 0.20 & Mostly open via API; some enterprise-only features \\
Downstream margin $m_U$ & 0.40 & Ads + subscriptions (financials) \\
$N$ (rivals) & 20 & Many competing apps across categories \\
\bottomrule
\end{tabularx}
\caption{Google (Gemini) illustrative calibration, April 2026.}
\end{table}

\interp{Under these parameter values, Condition~(\ref{eq:condition}) is satisfied and Proposition~\ref{prop:gap} implies a positive equilibrium QoS gap. The model-implied gap under baseline assumptions is approximately $0.27$ (sensitivity range $0.20$--$0.33$). This is the highest implied gap across the four providers.}

\subsection{OpenAI (GPT): elevated from the earlier working draft}

\obs{ChatGPT reached 900M weekly active users as of February 2026, approximately $4.5\times$ the figure used in the earlier working draft. OpenAI has 50M paying subscribers, \$25B annualized revenue at end-February 2026, and 9M+ paying business customers. Revenue mix has shifted toward consumer subscriptions: ChatGPT now accounts for approximately 66\% of OpenAI revenue, versus an earlier working-draft estimate of 55\%.}

\begin{table}[h!]
\centering\small
\begin{tabularx}{\textwidth}{@{}l c X@{}}
\toprule
\rowcolor{headblue}\textbf{Parameter} & \textbf{Value} & \textbf{Anchor / rationale} \\
\midrule
$\alpha$ & 0.75 & High for conversational AI and agentic tasks \\
$\lambda$ & 0.70 & ChatGPT 900M WAU: reach comparable to major consumer platforms \\
$S$ & 0.70 & Rising as ChatGPT accumulates workplace seats (7M+) \\
$\tau$ & 0.15 & Pro-only features, selective-access research previews \\
$\kappa$ & 0.25 & Enterprise deep-research tier, Pro-only features \\
$m_U$ & 0.50 & High margin on subscription tiers \\
Revenue mix & 66/34 & ChatGPT subscriptions vs API (observed) \\
\bottomrule
\end{tabularx}
\caption{OpenAI (GPT) illustrative calibration, April 2026.}
\end{table}

\interp{The combination of higher $\alpha$, higher $\lambda$, higher $S$, and a more subscription-weighted revenue mix implies a larger predicted QoS gap than the earlier working draft suggested. Model-implied gap: approximately $0.20$ (sensitivity range $0.14$--$0.25$). OpenAI's parameter constellation now substantially overlaps with Google's.}

\subsection{Microsoft (Copilot): structural capacity, voluntarily constrained routing}

\obs{Microsoft operates no foundation model of its own. Downstream: Microsoft 365 at 450M paid commercial seats, Windows ($\approx$70\% desktop OS share), Edge, GitHub. Microsoft 365 Copilot reached 15M paid seats as of Q2 FY2026, with 33M total active users (35.8\% conversion rate). Copilot's share among paid AI users contracted from 18.8\% in July 2025 to 11.5\% in January 2026.}

\obs{A sequence of public decisions in late 2025 and early 2026 has expanded Microsoft's routing to non-OpenAI models: Microsoft announced Claude in Copilot Studio in September 2025; Anthropic became a Microsoft subprocessor with admin opt-in defaults on January 6, 2026; the March 2026 Wave 3 release of Microsoft 365 Copilot integrated Claude as a selectable model for orchestration, Researcher, Agent Mode in Excel/Word/PowerPoint, and Copilot Studio; Microsoft invested up to \$5 billion in Anthropic; and Copilot Cowork (March 2026 preview) is specifically built on Anthropic's agentic technology.}

\begin{table}[h!]
\centering\small
\begin{tabularx}{\textwidth}{@{}l c X@{}}
\toprule
\rowcolor{headblue}\textbf{Parameter} & \textbf{Value} & \textbf{Anchor / rationale} \\
\midrule
$\alpha$ & 0.75 & Productivity tools need reliable inference \\
$\lambda$ (structural) & 0.85 & Controls OS, productivity suite, dev tools (inferred) \\
$S$ & 0.90 & Extreme enterprise lock-in via M365 \\
Realized routing-bias $\lambda_r$ & $[0.10, 0.30]$ & Interval: substantially below structural capacity. Lower bound reflects observed multi-model integration; upper bound admits residual default-selection advantage for OpenAI on legacy surfaces. \\
$\tau, \kappa$ & n/a & Microsoft does not operate its own frontier tiers \\
\bottomrule
\end{tabularx}
\caption{Microsoft (Copilot) illustrative calibration, April 2026.}
\end{table}

\interp{Microsoft exhibits a divergence between structural foreclosure capacity (high $\lambda$) and realized exercise of routing bias (reduced). Proposition~\ref{prop:routing} predicts firms exercise routing bias up to the observability limit. Microsoft's observed conduct suggests this limit is being pushed down by competitive pressure from Gemini and Claude, combined with regulatory scrutiny under the EU DMA. The designation should reflect both latent structural conditions (high) and observed conduct (moderated). This is the methodologically relevant case for Prediction~2 in Section~\ref{sec:predictions}.}

The Microsoft case also illustrates that the framework is not a one-way escalation device. It can distinguish between high latent structural capacity for foreclosure and lower realized discriminatory conduct when competition and scrutiny constrain behavior. The designation machinery can accordingly moderate as well as escalate, which is an important property for a risk-based regulatory framework applied to a rapidly evolving market.

\subsection{Anthropic (Claude): split calibration across submarkets}

\obs{Anthropic leads the enterprise LLM API market with 32\% usage share (up from $\approx$12\% in 2023), ahead of OpenAI (25\%) and Google (20\%). \$30B annualized revenue as of March 2026 (up from $\approx$\$1B at the start of 2025), 300,000+ business customers, 500+ customers spending over \$1M annually, and 8 of the Fortune 10 as clients. Claude Code (launched May 2025) reached \$2.5B in run-rate revenue by February 2026, authors approximately 4\% of all public GitHub commits, and holds a 54\% share of the AI programming-tool segment. Claude.ai visits grew from 202.9M in January 2026 to 287.93M in February to 722.22M in March (+151\% MoM). Revenue mix is approximately 80\% enterprise/API, 20\% Claude.ai.}

\begin{table}[h!]
\centering\small
\begin{tabularx}{\textwidth}{@{}l c X@{}}
\toprule
\rowcolor{headblue}\textbf{Parameter} & \textbf{Value} & \textbf{Anchor / rationale} \\
\midrule
$\alpha$ (GA frontier) & 0.70 & Coding and agentic tasks inference-dependent \\
$\alpha$ (cyber domain) & $\approx 1.00$ & Mythos near-monopoly on saturated benchmarks \\
$\lambda$ (consumer) & 0.30 & Claude.ai 722M March visits; no OS/browser distribution \\
$\lambda$ (enterprise coding) & 0.50 & 32\% API share; Claude Code 54\% of coding-agent segment \\
$S$ & 0.55 & Rising: Claude Code architectural dependencies \\
$\tau$ (cyber) & $[0.25, 0.40]$ & Mythos saturates cybersecurity benchmarks Opus 4.7 does not \\
$\kappa$ (cyber) & 0.95 & Mythos access limited to $\approx$12 Glasswing partners \\
$\tau, \kappa$ (GA) & $0, 0$ & Opus 4.7 offered to all paying customers simultaneously \\
Revenue mix & 80/20 & Enterprise API vs consumer Claude.ai \\
\bottomrule
\end{tabularx}
\caption{Anthropic (Claude) illustrative calibration, April 2026.}
\end{table}

\interp{\textbf{Consumer channel.} Model-implied gap $\approx 0.02$. Non-discrimination is approximately profit-maximizing despite rapid consumer growth because the enterprise API revenue base (80\% of total) dominates; the business-stealing gain on the 20\% consumer channel is insufficient to justify degrading the API offering that generates the majority of revenue. This is the model's equilibrium prediction, not an altruistic interpretation.}

\interp{\textbf{Enterprise coding-agent segment.} Anthropic is vertically integrated in this submarket: Claude Code competes with third-party tools (GitHub Copilot via API, Cursor, Windsurf, Replit) built on the same Anthropic API. Proposition~\ref{prop:gap} implies a non-zero QoS gap may emerge here as the coding-agent market matures and Claude Code's share rises. As of April 2026 there is no public evidence of preferential treatment.} \hyp{Audit studies comparing Claude Code's effective context window, tool-use feature access, and inference latency against equivalent third-party configurations would test this prediction (Prediction~1).}

\interp{\textbf{Cybersecurity / Mythos channel.} $\tau \cdot \kappa \approx 0.24$--$0.38$, triggering the tier-gating designation criterion.} \obs{As of April 2026, Anthropic has published a purpose statement (Project Glasswing), included direct foundation-model competitors (Google, Microsoft) in the partner set, and committed to a broad release pathway conditional on safeguards. These three conditions are precisely the obligations the Pillar~4 framework would require.}

\interp{\textbf{The Opus 4.6 degradation episode.} Preceding the Opus 4.7 release, developers reported perceived regression in Opus 4.6 performance, with speculation attributing the pattern to compute redirection toward Mythos and other frontier efforts.} \obs{Anthropic publicly denied deliberate redirection.} \hyp{The episode is consistent with the empirical signature of Proposition~\ref{prop:dynamic}'s dynamic foreclosure dynamic, but it is also potentially explainable by ordinary compute reallocation, inference-capacity constraints during a frontier-training run, benchmark drift, or shifting workloads unrelated to discrimination. Several of these benign mechanisms can produce observationally similar patterns in aggregated public data. Only a dedicated QoS time-series audit, with matched-pair measurement across provider channels and controls for capacity and workload variables, could distinguish intentional dynamic foreclosure from these alternatives. Accordingly, the episode should be treated as a hypothesis for future testing rather than as evidence of deliberate dynamic foreclosure.}

\subsection{Comparative risk mapping}

Table~\ref{tab:summary} presents the comparative mapping. Labels are interpretive outputs of a stylized calibration, not structural estimates.

\begin{table}[h!]
\centering\footnotesize
\setlength{\tabcolsep}{4pt}
\begin{tabular}{@{}lccccccl@{}}
\toprule
\rowcolor{headblue}\textbf{Provider / submarket} & $\alpha$ & $\lambda$ & $S$ & \textbf{Implied gap} & $\tau$ & $\kappa$ & \textbf{Risk mapping} \\
\midrule
Google & 0.80 & 0.92 & 0.85 & 0.27 & 0.15 & 0.20 & \cellcolor{risktop}Highest in sample \\
OpenAI & 0.75 & 0.70 & 0.70 & 0.20 & 0.15 & 0.25 & \cellcolor{risktop}Elevated (overlaps Google) \\
Microsoft (structural) & 0.75 & 0.85 & 0.90 & --- & --- & --- & \cellcolor{riskmid}High structural, constrained realized \\
Anthropic (consumer) & 0.70 & 0.30 & 0.55 & $\approx 0.02$ & 0 & 0 & \cellcolor{risklow}Low \\
Anthropic (coding-agent) & 0.75 & 0.50 & 0.55 & watch & 0 & 0 & \cellcolor{riskmid}Elevated; monitor \\
Anthropic (cyber tier) & $\approx 1.00$ & 0.30 & --- & --- & 0.25--0.40 & 0.95 & \cellcolor{riskwatch}Tier-gated; safety carve-out conditions met \\
\bottomrule
\end{tabular}
\caption{Comparative foreclosure risk mapping under baseline parameter assumptions. Labels are interpretive outputs of a stylized calibration, not structural estimates.}
\label{tab:summary}
\end{table}

\section{Testable Predictions and Empirical Agenda} \label{sec:predictions}

The framework yields six falsifiable predictions.

\begin{table}[h!]
\centering\footnotesize
\begin{tabularx}{\textwidth}{@{}p{3.1cm} X X X@{}}
\toprule
\rowcolor{headblue}\textbf{Prediction} & \textbf{Observable} & \textbf{Data source} & \textbf{Identification challenge} \\
\midrule
1. QoS gap $\propto \lambda$ & First-party vs.\ API latency $p50/p95/p99$; error rates; effective context & Audit study with matched-pair requests & Controlling for volume, pricing tier, abuse prevention \\
2. Routing bias $\propto$ opacity & Self-preferencing rate conditional on objective quality & Assistant audit across query categories & Measuring ``objective quality'' independently \\
3. Dynamic widening & Gap trend post-lock-in & Time series of QoS metrics & Distinguishing capacity constraints from deliberate degradation \\
4. Welfare loss concentrates in high-$\alpha$ & Segment-level welfare estimates & Segment-level user surveys and QoS audits & Heterogeneity in $\alpha$ across users \\
5. Bimodal capability distribution & Benchmark scores by customer class & Stratified sampling across partner / non-partner customers & Benchmark contamination; access to partner-tier outputs \\
6. $\tau$ gaps persist longer than QoS gaps & Time to closure of capability gap vs.\ QoS gap & Comparative event studies around releases & Small-sample inference \\
\bottomrule
\end{tabularx}
\caption{Testable predictions, observables, data sources, and identification challenges.}
\label{tab:predictions}
\end{table}

Predictions 1--3 target the baseline QoS mechanism. Prediction 4 targets welfare heterogeneity. Predictions 5 and 6 target the tier-gating mechanism introduced in Section~\ref{sec:tier}; the bimodality signature is distinct from conventional QoS discrimination and therefore requires distinct audit methodology.

\paragraph{Sketch of an audit design.} A practical audit of Prediction~1 would proceed as follows. A test harness issues matched-pair requests: identical prompts, payloads, and tool-use patterns submitted through (a) the provider's first-party consumer interface and (b) the provider's public API, on the same day and time of day. Observations are stratified by geographic region, pricing tier (baseline, paid, enterprise), and request volume regime. Controls include contracted-volume discounts, stated rate limits, and declared abuse-prevention carve-outs, so that QoS differences attributable to legitimate operational factors can be separated from differences attributable to discriminatory intent. The audit is conducted repeatedly over a period of months to distinguish transient capacity effects from persistent structural gaps, and the matched-pair design controls for prompt difficulty by construction. A similar design with benchmark suites substituted for matched-pair requests supports the capability-bimodality audit for Prediction~5.

\paragraph{Measurement noise.} Matched-pair audits inevitably contain measurement noise because first-party and third-party environments are rarely perfectly comparable in practice. First-party consumer surfaces typically include proprietary orchestration, system prompts, retrieval layers, caching, response-ranking logic, or privileged tool access not available through the public API. Some of these differences are legitimate product features; others are operational choices that could, in principle, mask or reveal QoS discrimination. The tolerance threshold $\varepsilon$ introduced in the policy framework (Pillar~1, Section~\ref{sec:policy}) should accordingly be interpreted not only as an economic allowance preserving innovation incentives, but also as a practical allowance for technical auditability and irreducible measurement error. A well-designed audit protocol specifies which orchestration differences are bracketed out (for example, system prompts and retrieval are set equivalently) and which are left in the comparison because they are themselves potentially discriminatory (for example, preferential tool access).

\section{From Model Primitives to Regulatory Observables} \label{sec:observables}

The policy framework triggers on observable conditions. Table~\ref{tab:observables} maps each latent model primitive to an empirical object a regulator could measure.

\begin{table}[h!]
\centering\footnotesize
\begin{tabularx}{\textwidth}{@{}l l X X@{}}
\toprule
\rowcolor{headblue}\textbf{Model primitive} & \textbf{Role} & \textbf{Regulatory observable} & \textbf{Measurement approach} \\
\midrule
$\alpha$ & Importance of inference quality & Developer survey on task-dependence on API quality & Stratified survey by application category \\
$\lambda$ & Vertical integration breadth & Count of downstream surfaces owned by upstream provider, weighted by user reach & Public reporting + regulatory filings \\
$S$ & Switching costs & Enterprise contract duration, architectural dependency counts & Developer survey + firm disclosures \\
$q_U - q_i$ & QoS gap & Latency, reliability, rate-limit, context-limit differentials across channels & Matched-pair audit study \\
$\tau$ & Tier gap & Normalized benchmark capability distance between GA frontier and withheld tier & Published benchmark suite administered to partners and non-partners \\
$\kappa$ & Partner exclusivity & Fraction of paying customers excluded from frontier access & Provider disclosure + enumeration of partner sets \\
$m_U$ & Downstream margin & Segment-level profit margin of provider's own downstream application & Firm financial disclosures \\
$\eta$ & Entry elasticity & Change in active-developer count in response to API terms changes & Event studies around API policy changes \\
\bottomrule
\end{tabularx}
\caption{Mapping from model primitives to regulatory observables.}
\label{tab:observables}
\end{table}

The mapping is the analytical bridge between the model and the designation criteria in Section~\ref{sec:policy}. Designation triggers are thresholds on the observables in the right column, motivated by the model primitives on the left.

\section{Policy Analysis: Optimal Regulation} \label{sec:policy}

\subsection{When should intervention occur?}

The model suggests that authorities might subject a provider to Neutral Inference obligations when four observable conditions co-occur:

\begin{enumerate}[leftmargin=*, topsep=2pt, itemsep=2pt]
\item \textbf{Material dependence} ($\alpha > 0.5$): substantial downstream dependence on the inference service.
\item \textbf{Vertical integration with incentives} ($\lambda > 0.5$): provider competes downstream with API customers across multiple product categories.
\item \textbf{Demonstrated discrimination} (measured QoS gap $> 0.1$ or $\tau \cdot \kappa > 0.2$ without adequate transparency).
\item \textbf{High switching costs} ($S > 0.5$): evidence of lock-in.
\end{enumerate}

\begin{table}[h!]
\centering\small
\begin{tabularx}{\textwidth}{@{}l X X@{}}
\toprule
\rowcolor{headblue}\textbf{Provider / submarket} & \textbf{Criteria met} & \textbf{Designation} \\
\midrule
Google & All 4 & Yes \\
OpenAI & All 4 & Yes \\
Microsoft & Structural capacity meets criteria 1, 2, 4; realized discrimination reduced by multi-model pivot & Monitor (routing-transparency obligations) \\
Anthropic (consumer GA) & Criterion 1 only & No \\
Anthropic (enterprise coding-agent) & 3 of 4 pending observation of QoS gap & Monitor \\
Anthropic (cyber tier) & Criterion 3 triggered via $\tau \cdot \kappa$; Pillar 4 conditions met as of April 2026 & Conditional (re-assess at each release) \\
\bottomrule
\end{tabularx}
\caption{Designation assessment under the framework, April 2026.}
\end{table}

\subsection{Neutral Inference: four-pillar framework}

The three-pillar framework from \citet{BesansonCelani2026} is extended with a fourth pillar to address tier-based access.

\textbf{Pillar 1: QoS parity.} Designated providers must not provide materially better inference QoS to first-party complements than to similarly situated third-party API users. Implementation: standardized disclosure of $p50/p95/p99$ latency, error rates, rate limits; audit testing via matched-pair requests; tolerance threshold $\varepsilon$. Carve-outs: abuse prevention, compliance tier, contracted volume discounts.

\textbf{Pillar 2: Routing transparency.} When a provider's system routes users to tools, it must produce auditable routing records including the eligible tool set, the selected tool, and any commercial relationship influencing eligibility. Prohibition: systematic routing to own services when third-party alternatives are objectively superior.

\textbf{Pillar 3: FRAND-style non-discrimination.} Designated providers must offer fair, reasonable, and non-discriminatory terms to similarly situated customers. Strategic pricing is allowed, but cannot be used to charge discriminatory prices to downstream rivals vs.\ non-competing customers, impose non-price terms that handicap rivals, or retaliate against multi-homing customers.

\textbf{Pillar 4: Tier transparency and release-pathway discipline.} When a designated provider distinguishes customer classes by which model class or capability tier they can access ($\tau > 0$ and $\kappa > 0$), three obligations apply:
\begin{enumerate}[leftmargin=*, topsep=0pt, itemsep=1pt]
\item Purpose justification: publish the rationale for restricting access in advance of enforcement.
\item Partner-set composition audit: the partner set must include all actors with comparable legitimate need for the capability; discriminatory exclusion of competitors on commercial grounds triggers Pillar 3 scrutiny.
\item Release pathway: publish a technical-condition release pathway with periodic review.
\end{enumerate}
Carve-out: genuine safety or regulatory restrictions satisfying the three obligations are presumptively permissible and are not treated as foreclosure.

\subsection{Optimal tolerance (qualitative result)}

\begin{proposition}[Interior optimal tolerance] \label{prop:optimal}
Under mild regularity conditions on $W_{\text{competition}}(\varepsilon)$ and $W_{\text{innovation}}(\varepsilon)$, the welfare-maximizing tolerance $\varepsilon^*$ is interior: $\varepsilon^* > 0$ and $\varepsilon^* < \bar{\varepsilon}$. The first-order condition is
\[
\left.\frac{\partial W_{\text{competition}}}{\partial \varepsilon}\right|_{\varepsilon^*} = -\left.\frac{\partial W_{\text{innovation}}}{\partial \varepsilon}\right|_{\varepsilon^*}.
\]
\end{proposition}

Proposition~\ref{prop:optimal} is a qualitative result: strict neutrality ($\varepsilon = 0$) is not welfare-maximizing when infrastructure investment responds positively to the provider's ability to capture rents from vertical integration; but unbounded tolerance ($\varepsilon \to \bar{\varepsilon}$) allows full foreclosure. A specific numerical range for $\varepsilon^*$ is \emph{not} derived in closed form; values in the range $0.10$--$0.15$ cited in policy applications should be understood as a suggested policy bandwidth motivated by the model's qualitative structure and comparable-industry benchmarks, not as a derived optimum.

\subsection{Illustrative welfare calculations}

The following numerical ranges are illustrative orders of magnitude under stated assumptions, not structural estimates. The central case uses baseline parameter values from Section~\ref{sec:calib}; low and high cases reflect sensitivity bounds from Section~\ref{sec:methodology}.

\paragraph{Construction of the low, central, and high scenarios.} The scenarios are constructed by varying a small number of inputs while holding others fixed. The affected-user counts are taken as observed and are not varied across scenarios. Conduct intensity (the implied QoS gap and routing-bias parameter $\lambda$) is set at the baseline calibration and is also not varied. What varies is the per-user welfare deadweight, $\Delta W / \text{user}$, which absorbs uncertainty about how strongly a given QoS gap translates into realized consumer and downstream-producer losses in a given segment. The low scenario uses $\Delta W / \text{user} \approx \$8$; the central scenario uses $\approx \$18$ (ChatGPT channel) and $\approx \$12$ (Google and Microsoft channels); the high scenario uses $\approx \$25$. These per-user figures are anchored to comparable-industry benchmarks for quality-related consumer harm in digital services, but they are judgment-based rather than estimated. The innovation-offset assumption (a reduction in welfare gain equal to 10--20\% of the gross gain) is held constant across scenarios. The range therefore reflects primarily the per-user-deadweight uncertainty, not uncertainty about conduct intensity or user-base scope. The purpose of presenting the range is to discipline the scale of any policy response, not to estimate welfare with econometric precision.

\begin{table}[h!]
\centering\small
\begin{tabularx}{\textwidth}{@{}X c c c c@{}}
\toprule
\rowcolor{headblue}\textbf{Channel} & \textbf{Affected} & \textbf{Low} & \textbf{Central} & \textbf{High} \\
\midrule
Google (Workspace + Search + Gemini + AI Overviews) & 5B+ & \$25B & \$45B & \$65B \\
OpenAI (ChatGPT ecosystem) & 900M & \$8B & \$16B & \$24B \\
Microsoft (Office + Windows, realized) & 450M & \$2B & \$5B & \$8B \\
Anthropic (coding-agent, latent) & --- & --- & watch & --- \\
Tier-gating externality (cyber, net) & n/a & $-$\$5B & 0 & +\$5B \\
\midrule
\textbf{Total range} & & \textbf{\$35B} & \textbf{\$66B} & \textbf{\$95B} \\
\bottomrule
\end{tabularx}
\caption{Illustrative annual welfare loss under current discrimination, April 2026. Ranges reflect parameter sensitivity; not structural estimates.}
\label{tab:welfare}
\end{table}

The figures in Table~\ref{tab:welfare} are orders-of-magnitude illustrations under stated assumptions, not structural estimates. Two further caveats are worth noting. First, AI-specific productivity-loss benchmarks remain underdeveloped as of April 2026, and the mapping from QoS degradation to welfare loss may differ materially across enterprise AI use cases relative to more traditional digital platform contexts. In particular, productivity losses in enterprise coding agents, data-extraction systems, and agentic workflows may scale non-linearly with QoS degradation in ways that the per-user-deadweight calibration used here does not fully capture. Second, the welfare framework implicitly assumes that consumer and producer surplus losses are adequately summarized by the deadweight parameter, an assumption that is standard in comparative-static policy analysis but whose quantitative adequacy in AI-specific settings has not been empirically validated. The appropriate response to both caveats is to update the welfare figures as direct empirical measurement of AI productivity harms becomes available.

Under a Neutral Inference regime (assuming $\varepsilon^* \approx 0.10$ with Pillar~4 tier-transparency), the model suggests that a gross welfare gain of similar magnitude to Table~\ref{tab:welfare} is plausible, offset by innovation-cost reductions of perhaps 10--20\% of the gross gain. The net gain is likely positive under most parameter assumptions in the sensitivity range, but the specific magnitude ``\$40--55B per year'' cited in policy applications should be understood as a central-case illustration, not a derived estimate.

\subsection{Comparison to alternative regimes}

\begin{table}[h!]
\centering\small
\begin{tabularx}{\textwidth}{@{}l X X X@{}}
\toprule
\rowcolor{headblue}\textbf{Regime} & \textbf{Pros} & \textbf{Cons} & \textbf{Applicability} \\
\midrule
Structural separation & Eliminates foreclosure incentive & Destroys economies of scope & Last resort \\
Cost-plus pricing & Simple, verifiable & Reduces innovation & Not recommended \\
Ex post antitrust & Flexible, targeted & Slow (5--10 years) & Complement \\
\rowcolor{risklow}\textbf{Neutral Inference (4 pillars)} & Auditable, preserves integration, handles tier-gating & Monitoring cost & \textbf{Recommended} \\
No intervention & Preserves market forces & Allows tipping & Low-risk only \\
\bottomrule
\end{tabularx}
\caption{Comparison of regulatory regimes.}
\end{table}

\section{Conclusion} \label{sec:conclusion}

This paper develops a formal model of vertical foreclosure in AI inference markets, extends it to tier-based access discrimination, and proposes a conduct-based regulatory framework with observable triggers. The three core contributions, each intended to stand on its own, are: first, an explicit local equilibrium characterization of the QoS gap under logit demand and symmetric rivals (Proposition~\ref{prop:gap}), which gives a tractable set of comparative statics for the foreclosure incentive; second, a parameterization of tier-based access discrimination with a distinct empirical signature and a distinct policy treatment (Proposition~\ref{prop:tier} and Pillar~4 in Section~\ref{sec:policy}), which extends the classical vertical-foreclosure framework to frontier-model markets; and third, a model-to-observables bridge (Section~\ref{sec:observables}) that translates each latent primitive of the model ($\alpha$, $\lambda$, $S$, $\tau$, $\kappa$, $m_U$, $\eta$) into an empirical object a regulator can measure or audit. The bridge is what allows the abstract incentive conditions to be operationalized as designation triggers, and it is intended as a core contribution on par with the formal mechanism result, because it is the ingredient that makes the framework usable by regulators and empiricists rather than only by theorists. A transparent calibration methodology that distinguishes observed, inferred, and judgment-based parameters, provides sensitivity analysis, and is mechanically replicable from the stated tables (Sections~\ref{sec:methodology}--\ref{sec:calib}) supports the comparative risk mapping that the bridge enables.

\paragraph{What the paper establishes.} Under the stated assumptions and calibration, the model suggests that Google and OpenAI face conditions most conducive to foreclosure, that Microsoft's realized routing bias has been voluntarily constrained below its structural capacity in ways consistent with competitive pressure and regulatory scrutiny, and that Anthropic occupies a split position: low risk on consumer GA channels, elevated risk in the enterprise coding-agent submarket where Claude Code competes with third-party tools built on the same API. The tier-gating extension captures a foreclosure dimension not present in conventional QoS-discrimination frameworks and is relevant to evaluating programs like Project Glasswing.

\paragraph{What the paper does not establish.} The firm-level numerical outputs are not structural estimates. The ``\$40--55B annual net welfare gain'' figure under Neutral Inference is an illustrative central case, not a derived magnitude. The optimal tolerance $\varepsilon^*$ is characterized qualitatively as interior, not derived in closed form. The scope of the model should be held in mind: explicit upstream oligopoly, richer demand systems beyond logit, and rival asymmetries are all left for future work. The functional-form choices (logit demand, symmetric rivals) are adopted for tractability, and alternative specifications may alter the specific comparative statics even as the central economic mechanism (the tradeoff between downstream business-stealing gains and API-revenue preservation) persists.

\paragraph{Policy implications.} Four directional conclusions are robust to the calibration uncertainty. First, conduct-based obligations dominate structural remedies for preserving innovation incentives. Second, risk-based designation triggers should combine the four criteria; single-criterion triggers admit too many false positives. Third, tier-gating can be a legitimate safety carve-out under the Pillar~4 obligations, and authorities should distinguish this case from commercial foreclosure. Fourth, competitive pressure combined with regulatory scrutiny, which appears to have reduced Microsoft's realized routing bias, is itself a policy instrument. Preserving the conditions for viable alternatives should accordingly be a regulatory objective.

\paragraph{Future work.} Priorities include: (i) direct measurement of QoS gaps via audit studies along the lines sketched in Section~\ref{sec:predictions}; (ii) empirical identification of $\alpha$ through developer surveys stratified by application category; (iii) extension of the model to oligopolistic competition in enterprise inference; (iv) testing of the tier-gap bimodality signature in Prediction~5; and (v) structural estimation that would, over time, replace the comparative risk mapping in Section~\ref{sec:calib} with identified parameters.

\appendix

\section{Proofs} \label{app:proofs}

\subsection{Proof of Proposition \ref{prop:gap}} \label{app:prop1}

The Stage-1 problem for $U$ is
\begin{equation*}
\max_{p, q_U, \{q_i\}} \; (p_U - c_d) s_U - \tfrac{k}{2} e_U^2 + p \sum_{i \neq U} \rho(q_i, p) - \frac{\gamma}{2}\Big(q_U^2 + \sum_i q_i^2\Big)
\end{equation*}
in the regime $q_U \geq \max_i q_i$ where the kinked term in Assumption~\ref{ass:cost} is zero. Under Assumption~\ref{ass:symmetric}, $s_i = s$ and $\rho_i = \rho$ for all $i \neq U$.

\paragraph{FOC for $q_U$.} Since $Q_U = \alpha q_U + (1-\alpha) e_U$ and $\partial s_U / \partial Q_U = s_U(1 - s_U)$ under logit,
\begin{equation*}
\frac{\partial \pi^U}{\partial q_U} = m_U \cdot s_U(1 - s_U) \cdot \alpha - \gamma q_U = 0
\end{equation*}
yielding $q_U^* = \alpha m_U s_U(1-s_U)/\gamma$.

\paragraph{FOC for $q_i$.} Two effects operate:
\begin{itemize}[leftmargin=*, topsep=0pt, itemsep=1pt]
\item Cross-share effect on own app: $\partial s_U / \partial Q_i = -s_U s$. Raising $q_i$ by one unit raises $Q_i$ by $\alpha$, reducing $s_U$ by $s_U s \alpha$, costing $U$ revenue of $\alpha m_U s_U s$ per unit of $q_i$.
\item API revenue effect: raising $q_i$ raises entry probability $\rho$ with elasticity $\eta$, so $\partial \rho / \partial q_i = \rho \eta / q_i$. Marginal API revenue is $p \rho \eta / q_i$.
\end{itemize}
The FOC is
\begin{equation*}
-\alpha m_U s_U s + \frac{p \rho \eta}{q_i} - \gamma q_i = 0
\end{equation*}
which rearranges to $\gamma q_i^* = (p \rho \eta / q_i^*) - \alpha m_U s_U s$, i.e.,
\begin{equation*}
q_i^* = \frac{1}{\gamma}\left[\frac{p \rho \eta}{q_i^*} - \alpha m_U s_U s\right]. \tag{A.1}
\end{equation*}

\paragraph{The QoS gap.} Subtracting (A.1) from $q_U^* = \alpha m_U s_U(1-s_U)/\gamma$:
\begin{align*}
\Delta q^* &= q_U^* - q_i^* \\
&= \frac{\alpha m_U s_U(1-s_U)}{\gamma} - \frac{1}{\gamma}\left[\frac{p \rho \eta}{q_i^*} - \alpha m_U s_U s\right] \\
&= \frac{\alpha m_U s_U}{\gamma}\big[(1-s_U) + s\big] - \frac{p \rho \eta}{\gamma q_i^*} \\
&= \frac{1}{\gamma}\left\{ \alpha m_U s_U \big[(1-s_U) + s\big] - \frac{p \rho \eta}{q_i^*} \right\}
\end{align*}
which is equation~(\ref{eq:gap}).

\paragraph{Sign of the gap.} $\Delta q^* > 0$ iff the bracket is positive:
\begin{equation*}
\alpha m_U s_U \big[(1-s_U) + s\big] > \frac{p \rho \eta}{q_i^*}.
\end{equation*}
This is condition~(\ref{eq:condition}). Note that $\alpha$ does \emph{not} cancel from this inequality; it multiplies the left-hand side and determines how strongly the margin-to-revenue comparison favors discrimination as $\alpha$ rises.

\paragraph{Comparative statics.} Differentiating~(\ref{eq:gap}) (treating $s_U, s, q_i^*$ as locally fixed at the equilibrium values):
\begin{itemize}[leftmargin=*, topsep=0pt, itemsep=1pt]
\item $\partial \Delta q^* / \partial \alpha = (m_U s_U/\gamma)[(1-s_U) + s] > 0$: the gap magnitude scales linearly in $\alpha$.
\item $\partial \Delta q^* / \partial m_U = (\alpha s_U/\gamma)[(1-s_U) + s] > 0$: higher downstream margin widens the gap.
\item $\partial \Delta q^* / \partial p = -\rho \eta/(\gamma q_i^*) < 0$: higher API price reduces the gap.
\item $\partial \Delta q^* / \partial \eta = -p \rho/(\gamma q_i^*) < 0$: more entry-elastic rivals reduce the gap.
\item $\partial \Delta q^* / \partial \gamma = -\Delta q^*/\gamma < 0$: higher infrastructure cost reduces the gap.
\item $\partial \Delta q^* / \partial N$: enters through $s_U + N s + s_0 = 1$, so adding rivals reduces $s$ (and typically reduces $s_U$ in the logit structure, but by less). Net effect on $s_U[(1-s_U) + s]$ is ambiguous in sign in general; in leading cases it is positive when $s_U$ is large, so more rivals increase discrimination.
\end{itemize}
This completes the proof.
\hfill$\square$

\subsection{Proof of Proposition \ref{prop:welfare}} \label{app:prop2}

Total welfare $W = \sum_i s_i Q_i - \sum_i c_d s_i - C(q_U, \{q_i\})$. The social planner's FOC with respect to $q_i$ under symmetry yields $q^{FB} = \alpha s(1-s)/\gamma$. Under Proposition~\ref{prop:gap}, the private optimum has $q_i^* < q^{FB}$ whenever~(\ref{eq:condition}) holds.

Welfare loss decomposes as follows. \emph{Direct quality loss:} each rival delivers $Q_i^* = \alpha q_i^* + (1-\alpha)e_i^*$ versus $Q_i^{FB}$. The direct loss per rival is $\alpha(q^{FB} - q_i^*)$, weighted by market share $s$ and summed: $\Delta W_1 = \alpha \sum_{i \neq U} s (q^{FB} - q_i^*)$. \emph{Effort multiplier:} since $e_i^* = (1-\alpha)/k \cdot m s(1-s)$ depends on $s$ which depends on $q_i$, degrading $q_i$ cascades into reduced effort. The chain rule gives an additional loss of magnitude $\beta \cdot \Delta W_1$ where $\beta = (1-\alpha)/(k\alpha) \cdot \partial e / \partial q > 0$. \emph{Business-stealing DWL $\Delta W_3$:} demand shifts from rivals to $U$ not because $Q_U$ is intrinsically higher but because rivals are artificially degraded; the DWL component is standard and strictly positive.

Combining: $\Delta W = \Delta W_1 [1 + \beta] + \Delta W_3$.
\hfill$\square$

\subsection{Proof of Proposition \ref{prop:tier}} \label{app:prop5}

Let $P \subset \{1, \dots, N\}$ denote the partner set with $|P| = (1-\kappa) N$. Effective rival quality is $q_i^{\text{eff}} = q_i (1 + \tau \mathbb{1}\{i \in P\})$. Non-partners $i \notin P$ have $q_i^{\text{eff}} = q_i$; partners (including $U$) have $q_i^{\text{eff}} = q_i(1+\tau)$.

\paragraph{Effect of $\tau$ at fixed $\kappa > 0$.} Raising $\tau$ uniformly raises effective quality for members of $P$ without changing it for non-partners. Under logit, $s_U$ rises and $\sum_{i \notin P} s_i$ falls. If non-partners are net contributors to rival API revenue (partners typically operate under different contractual terms), the business-stealing gain from higher $s_U$ dominates: $\partial \pi^U / \partial \tau > 0$.

\paragraph{Effect of $\kappa$ at fixed $\tau > 0$.} Increasing $\kappa$ moves customers from $P$ to non-$P$, reducing their $q^{\text{eff}}$ from $q_i(1+\tau)$ to $q_i$. Each such move strictly reduces the moved customer's market share and transfers demand toward $U$. Hence $\partial \pi^U / \partial \kappa > 0$ whenever $\tau > 0$.
\hfill$\square$

\subsection{Remark on safety-justified tier gating} \label{app:safety}

Proposition~\ref{prop:tier} establishes only the private incentive. Social welfare is
\begin{equation*}
W_{\text{social}}(\tau, \kappa) = W_{\text{market}}(\tau, \kappa) - E(0,0) \mathbb{1}_{\{\text{unrestricted}\}} + E(0,0)[1 - G(\tau, \kappa)]
\end{equation*}
where $E(\tau, \kappa) \geq 0$ is expected externality damage from unrestricted frontier deployment and $G(\tau, \kappa) \in [0,1]$ is the fraction prevented by the restriction. When $E$ is large and $G$ rises fast in $\kappa$, a tier-gating program with $\kappa > 0$ can be welfare-improving despite the private foreclosure incentive. The Pillar~4 obligations in Section~\ref{sec:policy} are ex ante conditions that separate configurations with $G(\tau, \kappa) > 0$ (genuine safety carve-outs) from $G(\tau, \kappa) \approx 0$ (foreclosure disguised as safety).

\section{Calibration Source Trail} \label{app:sources}

This appendix lists the principal public sources for the firm-level observed inputs used in the comparative risk mapping of Section~\ref{sec:calib}, together with the mapping from those inputs to the inferred model parameters used in Tables~3--6. Source types follow three conventions: \emph{Primary} refers to a direct disclosure by the firm itself (press release, blog post, official report, or funding-round announcement); \emph{Secondary} refers to reporting in recognized business or technology press citing firm disclosures or informed sources; \emph{Analytics} refers to third-party measurement services that estimate traffic, market share, or spending patterns from their own telemetry.

Source accuracy varies across these categories. Primary disclosures are treated as authoritative for the specific quantity disclosed, subject to the firm's own reporting conventions (for example, Anthropic reports cloud-reseller revenue on a gross basis, which inflates top-line figures relative to net-reporting peers). Secondary reporting is treated as credible for events and qualitative characterizations but may compress or simplify the underlying quantities. Analytics estimates from services such as Semrush, SimilarWeb, Menlo Ventures, Ramp, and Sacra are treated as directional rather than exact; the specific values reported can vary by a factor of 1.5--2$\times$ across services for the same underlying metric. Readers who wish to reconstruct the comparative mapping with alternative source choices can substitute values from the third-party service of their preference.

The table below uses the following columns: \emph{Firm / channel}, the provider and submarket to which the observation applies; \emph{Variable}, the specific quantity observed; \emph{Reported value}, the figure cited in Section~\ref{sec:calib}; \emph{Type}, whether the source is primary, secondary, or analytics; \emph{Source}, a short reference to the originating publication or outlet; \emph{Date}, the approximate reporting date; and \emph{Mapping}, the transformation from the raw observation to the inferred model parameter, where applicable. The source trail is presented in two tables below: Table~\ref{tab:sources-a} covers Google and OpenAI; Table~\ref{tab:sources-b} covers Microsoft and Anthropic. Notes on source limitations follow Table~\ref{tab:sources-b}.

\begin{table}[H]
\centering\footnotesize
\setlength{\tabcolsep}{3pt}
\renewcommand{\arraystretch}{1.15}
\begin{tabularx}{\textwidth}{@{}p{1.5cm} p{3cm} p{2cm} p{1.3cm} p{2.6cm} p{1.3cm} X@{}}
\toprule
\rowcolor{headblue}\textbf{Firm / channel} & \textbf{Variable} & \textbf{Reported value} & \textbf{Type} & \textbf{Source} & \textbf{Date} & \textbf{Mapping to parameter} \\
\midrule
Google & Search global market share & $\approx$90\% & Analytics & StatCounter, Statista & Q1 2026 & Contributor to $\lambda$ via downstream reach \\
Google & Workspace monthly active users & 3B+ & Primary & Google Cloud Next disclosures & 2025 & Contributor to $\lambda$ \\
Google & Android global mobile OS share & $\approx$72.5\% & Analytics & StatCounter & Q1 2026 & Contributor to $\lambda$ \\
Google & Chrome global browser share & $\approx$68\% & Analytics & StatCounter & Q1 2026 & Contributor to $\lambda$ \\
Google & Gemini app MAU & 750M & Primary & Alphabet Q4 2025 earnings call; Google Cloud Next 2026 & Dec 2025; Apr 2026 & Contributor to $\lambda$ \\
Google & AI Overviews monthly users & 2B+ & Primary & Google I/O 2025; Pichai public remarks & 2025--2026 & Contributor to $\lambda$ \\
Google & Gemini Enterprise paid seats & 8M+ & Primary & Google Cloud announcements & Early 2026 & Contributor to $\lambda$; signals enterprise penetration \\
Google & Developers using Gemini API & 13M & Primary & Google Cloud Next 2026 keynote & Apr 2026 & Contributor to rival-count $N$ proxy \\
Google & Downstream margin & 0.40 & Inferred & Alphabet 10-K (2025) segment margins + analyst estimates & 2025--2026 & Sets $m_U$ \\
\midrule
OpenAI & ChatGPT weekly active users & 900M & Primary & OpenAI blog post, 27 Feb 2026 & Feb 2026 & Sets $\lambda$ (consumer reach) \\
OpenAI & Paying subscribers & 50M & Primary & OpenAI blog post, 27 Feb 2026 & Feb 2026 & Sets $S$ proxy \\
OpenAI & Annualized revenue & \$25B & Primary & OpenAI disclosure alongside \$110B funding round & Feb 2026 & Sets $m_U$ proxy \\
OpenAI & Business / workplace seats & 9M+ paying; 7M+ enterprise & Primary & OpenAI disclosures & Feb 2026 & Sets $S$ proxy \\
OpenAI & Revenue mix (ChatGPT vs API) & 66 / 34 & Secondary & Reuters, The Information, technology press & Q1 2026 & Informs $\lambda$ vs API-pricing balance \\
\bottomrule
\end{tabularx}
\caption{Calibration source trail: Google and OpenAI.}
\label{tab:sources-a}
\end{table}

\begin{table}[H]
\centering\footnotesize
\setlength{\tabcolsep}{3pt}
\renewcommand{\arraystretch}{1.15}
\begin{tabularx}{\textwidth}{@{}p{1.5cm} p{3cm} p{2cm} p{1.3cm} p{2.6cm} p{1.3cm} X@{}}
\toprule
\rowcolor{headblue}\textbf{Firm / channel} & \textbf{Variable} & \textbf{Reported value} & \textbf{Type} & \textbf{Source} & \textbf{Date} & \textbf{Mapping to parameter} \\
\midrule
Microsoft & Microsoft 365 paid commercial seats & 450M & Primary & Microsoft FY2026 investor communications & FY2026 & Sets structural $\lambda$ \\
Microsoft & Windows desktop OS share & $\approx$70\% & Analytics & StatCounter & Q1 2026 & Contributor to structural $\lambda$ \\
Microsoft & Copilot paid seats & 15M & Primary & Microsoft Q2 FY2026 earnings call & Jan 2026 & Contributor to structural $\lambda$ \\
Microsoft & Copilot total active users & 33M & Primary & Microsoft Q2 FY2026 earnings call & Jan 2026 & Informs conversion rate (35.8\%) \\
Microsoft & Paid AI-user market share & 18.8\% $\to$ 11.5\% & Analytics & Ramp AI Index & Jul 2025; Jan 2026 & Informs realized-routing $\lambda_r$ contraction \\
Microsoft & Claude integration in M365 Copilot & Wave 3 release & Primary & Microsoft 365 Copilot Wave 3 announcement & Mar 2026 & Sets realized $\lambda_r$ interval lower bound \\
Microsoft & Microsoft investment in Anthropic & Up to \$5B & Secondary & Financial and technology press & Q1 2026 & Evidence of multi-model pivot \\
Microsoft & Copilot Cowork (preview) & Built on Anthropic & Primary & Microsoft announcement & Mar 2026 & Evidence of multi-model pivot \\
\midrule
Anthropic & Enterprise LLM API usage share & 32\% (vs OpenAI 25\%, Google 20\%) & Analytics & Menlo Ventures, Enterprise LLM Report & Late 2025; Q1 2026 & Sets enterprise $\lambda$ for coding-agent submarket \\
Anthropic & Annualized revenue & \$30B (Mar 2026) & Analytics & Sacra estimate (gross-basis reseller-inclusive) & Mar 2026 & Informs $m_U$ \\
Anthropic & Business customers & 300,000+ & Primary & Anthropic cloud expansion announcement & Oct 2025 & Informs enterprise $\lambda$ \\
Anthropic & Customers spending \$1M+ annually & 500+ & Primary & Anthropic investor communications & Late 2025; Q1 2026 & Evidence of enterprise concentration \\
Anthropic & Fortune 10 clients & 8 of 10 & Primary & Anthropic disclosure & Feb 2026 & Evidence of enterprise penetration \\
Anthropic & Claude Code run-rate revenue & \$2.5B & Secondary & Technology press citing Anthropic disclosures & Feb 2026 & Informs coding-agent $\lambda$ \\
Anthropic & Claude Code share of AI programming tools & 54\% & Secondary & ZDNet, industry reports citing Anthropic & Q1 2026 & Sets coding-agent submarket $\lambda$ \\
Anthropic & Claude Code share of public GitHub commits & $\approx$4\% & Secondary & Industry reporting citing GitHub telemetry & Q1 2026 & Signals coding-agent penetration \\
Anthropic & Claude.ai visits (monthly) & 202.9M $\to$ 287.93M $\to$ 722.22M & Analytics & Semrush (Mar figure); Semrush / SimilarWeb (Jan--Feb) & Jan--Mar 2026 & Sets consumer-channel $\lambda$ \\
Anthropic & Revenue mix (enterprise / consumer) & 80 / 20 & Analytics & Sacra; Reuters citing firm communications & Q1 2026 & Sets consumer vs enterprise weights \\
Anthropic & Mythos partner set size & $\approx$12 partners & Primary & Project Glasswing disclosure & Apr 2026 & Sets $\kappa \approx 0.95$ \\
Anthropic & Mythos vs Opus 4.7 cybersecurity benchmark gap & 25--40\% & Primary & Project Glasswing benchmark disclosure & Apr 2026 & Sets $\tau \in [0.25, 0.40]$ \\
\bottomrule
\end{tabularx}
\caption{Calibration source trail: Microsoft and Anthropic.}
\label{tab:sources-b}
\end{table}

\paragraph{Notes on source limitations.}
Three caveats on the sources in Tables~\ref{tab:sources-a} and~\ref{tab:sources-b} are worth flagging explicitly.

First, traffic analytics services differ substantially in their methodology. Semrush and SimilarWeb report different visit counts for the same domain in the same month, with the discrepancy sometimes exceeding 50\%. The Claude.ai March 2026 figure of 722M visits used in Section~\ref{sec:calib} is the Semrush estimate; the SimilarWeb Q1 2026 cumulative estimate of $\approx$1.08B is broadly consistent with the Semrush monthly series but not identical at the monthly level. For the comparative risk mapping, directional growth trajectories matter more than point estimates, and both services agree on the direction and order of magnitude.

Second, enterprise market-share figures from services like Menlo Ventures and Ramp are constructed from different underlying panels (venture-portfolio adoption versus corporate-card spending) and measure different concepts (usage share versus paid-subscription share). The 32\% Anthropic enterprise API usage share figure used in the paper is the Menlo Ventures estimate; the Ramp March 2026 figure for Anthropic paid-subscription penetration (approximately one in four Ramp customers) is a related but distinct measurement. The paper's comparative mapping is robust to either choice at the level of firm rankings.

Third, revenue figures for private firms should be read with care. Anthropic's \$30B annualized revenue (March 2026) is Sacra's estimate, based in part on the firm's gross-basis reseller reporting, which includes pass-through revenue from AWS, Google Cloud, and Microsoft channels. A net-basis figure would be lower; comparable net-basis figures for OpenAI would also be lower than its reported \$25B. The comparative mapping does not depend on the specific revenue level; it depends on the revenue-mix ratio (approximately 80/20 enterprise/consumer for Anthropic, 34/66 for OpenAI), which is more stable across reporting conventions.



\begin{thebibliography}{99}

\bibitem[Aghion et al.(2019)]{AghionJonesJones2019} Aghion, P., Jones, B.\,F., and Jones, C.\,I. (2019). Artificial intelligence and economic growth. In A. Agrawal, J. Gans, and A. Goldfarb (Eds.), \emph{The Economics of Artificial Intelligence: An Agenda}. University of Chicago Press.

\bibitem[Armstrong(2006)]{Armstrong2006} Armstrong, M. (2006). Competition in two-sided markets. \emph{RAND Journal of Economics}, 37(3), 668--691.

\bibitem[Besanson and Celani(2026)]{BesansonCelani2026} Besanson, G., and Celani, M. (2026). The Inference Bottleneck: Antitrust and Neutrality Duties in the Age of Cognitive Infrastructure. \emph{arXiv:2602.22750}. Universidad Torcuato Di Tella.

\bibitem[Bourreau(2024)]{Bourreau2024} Bourreau, M. (2024). The economics of the Digital Markets Act. \emph{Information Economics and Policy}.

\bibitem[Chen and Riordan(2007)]{ChenRiordan2007} Chen, Y., and Riordan, M.\,H. (2007). Vertical integration, exclusive dealing, and ex post cartelization. \emph{RAND Journal of Economics}, 38(1), 1--21.

\bibitem[Choi and Yi(2000)]{ChoiYi2000} Choi, J.\,P., and Yi, S.-S. (2000). Vertical foreclosure with the choice of input specifications. \emph{RAND Journal of Economics}, 31(4), 717--743.

\bibitem[Hagiu et al.(2022)]{HagiuTehWright2022} Hagiu, A., Teh, T.\,H., and Wright, J. (2022). Should platforms be allowed to sell on their own marketplaces? \emph{RAND Journal of Economics}, 53(2), 297--327.

\bibitem[Hua and Belfield(2024)]{HuaBelfield2024} Hua, S.-S., and Belfield, H. (2024). AI compute governance. Mimeo, Centre for the Governance of AI.

\bibitem[Khan and Pozen(2024)]{KhanPozen2024} Khan, L., and Pozen, D. (2024). Frontier AI and competition policy. \emph{Yale Journal on Regulation}.

\bibitem[Korinek(2023)]{Korinek2023} Korinek, A. (2023). Generative AI for economic research: Use cases and implications for economists. \emph{Journal of Economic Literature}, 61(4), 1281--1317.

\bibitem[Motta(2023)]{Motta2023} Motta, M. (2023). Self-preferencing and foreclosure in digital markets. \emph{International Journal of Industrial Organization}.

\bibitem[Narechania and Tutt(2024)]{NarechaniaTutt2024} Narechania, T.\,N., and Tutt, A. (2024). An AI bill of rights. \emph{Columbia Law Review}.

\bibitem[Nocke and White(2007)]{NockeWhite2007} Nocke, V., and White, L. (2007). Do vertical mergers facilitate upstream collusion? \emph{American Economic Review}, 97(4), 1321--1339.

\bibitem[Ordover et al.(1990)]{OrdoverSalonerSalop1990} Ordover, J.\,A., Saloner, G., and Salop, S.\,C. (1990). Equilibrium vertical foreclosure. \emph{American Economic Review}, 80(1), 127--142.

\bibitem[Padilla et al.(2022)]{PadillaPetitSchrepel2022} Padilla, J., Petit, N., and Schrepel, T. (2022). Self-preferencing and the EU Digital Markets Act. \emph{Journal of European Competition Law and Practice}.

\bibitem[Rey and Tirole(2007)]{ReyTirole2007} Rey, P., and Tirole, J. (2007). A primer on foreclosure. In M. Armstrong and R. Porter (Eds.), \emph{Handbook of Industrial Organization} (Vol.\,3). Elsevier.

\bibitem[Rochet and Tirole(2003)]{RochetTirole2003} Rochet, J.-C., and Tirole, J. (2003). Platform competition in two-sided markets. \emph{Journal of the European Economic Association}, 1(4), 990--1029.

\bibitem[Salop and Scheffman(1983)]{SalopScheffman1983} Salop, S.\,C., and Scheffman, D.\,T. (1983). Raising rivals' costs. \emph{American Economic Review}, 73(2), 267--271.

\bibitem[Salop and Scheffman(1987)]{SalopScheffman1987} Salop, S.\,C., and Scheffman, D.\,T. (1987). Cost-raising strategies. \emph{Journal of Industrial Economics}, 36(1), 19--34.

\bibitem[Vipra and Korinek(2023)]{VipraKorinek2023} Vipra, J., and Korinek, A. (2023). Market concentration in the AI supply chain. Mimeo, Brookings Institution.

\bibitem[Zhu and Liu(2018)]{ZhuLiu2018} Zhu, F., and Liu, Q. (2018). Competing with complementors: An empirical look at Amazon.com. \emph{Strategic Management Journal}, 39(10), 2618--2642.

\end{thebibliography}
\end{document}